\documentclass[runningheads]{llncs}

\usepackage[T1]{fontenc}
\usepackage[utf8]{inputenc}
\usepackage{graphicx}
\usepackage{amsmath,amssymb}
\usepackage{stmaryrd}
\usepackage{xcolor}
\usepackage{listings}
\usepackage{url}
\usepackage{hyperref}
\usepackage{booktabs}
\usepackage{multirow}

\newcommand{\infer}[3][]{\dfrac{#2}{#3}\;{\scriptstyle\mathsf{#1}}}

\lstdefinelanguage{Lean}{
  keywords={def, theorem, lemma, inductive, structure, class, instance, where,
            if, then, else, match, with, let, in, have, show, by, fun,
            namespace, open, variable, axiom, example, import, Prop, Type, Sort},
  keywordstyle=\color{blue}\bfseries,
  ndkeywords={Nat, Bool, Term, Star, Plus, Joinable, Diamond, Confluent,
              Terminating, LocalConfluent, ParRed, BetaStep, HasType, SN, Ty},
  ndkeywordstyle=\color{purple},
  sensitive=true,
  comment=[l]{--},
  morecomment=[s]{/-}{-/},
  commentstyle=\color{gray}\itshape,
  stringstyle=\color{red},
  morestring=[b]"
}

\newcommand{\lean}[1]{\lstinline[language=Lean]|#1|}
\newcommand{\redarrow}{\rightarrow}
\newcommand{\redstar}{\rightarrow^*}
\newcommand{\parred}{\Rightarrow}
\newcommand{\join}{\downarrow}

\begin{document}

\title{A Modular Lean~4 Framework for Confluence and Strong Normalization of Lambda Calculi with Products and Sums}

\titlerunning{Confluence and Normalization in Lean 4}

\author{Arthur F. Ramos\inst{1} \and
        Anjolina G. de Oliveira\inst{2} \and
        Ruy J. G. B. de Queiroz\inst{2} \and
        Tiago M. L. de Veras\inst{3}}

\authorrunning{A. F. Ramos et al.}

\institute{
Microsoft, Tampa, FL, USA\\
\email{arfreita@microsoft.com}
\and
Centro de Inform\'{a}tica, Universidade Federal de Pernambuco,\\
Recife, PE, Brazil\\
\email{\{ago,ruy\}@cin.ufpe.br}
\and
Departamento de Matem\'{a}tica, Universidade Federal Rural de Pernambuco,\\
Recife, PE, Brazil\\
\email{tiago.veras@ufrpe.br}
}

\maketitle

\begin{abstract}
We present \textsc{Metatheory}, a comprehensive library for programming language
foundations in Lean~4, featuring a modular framework for proving confluence of
abstract rewriting systems. The library implements three classical proof techniques---the
diamond property via parallel reduction, Newman's lemma for terminating systems,
and the Hindley-Rosen lemma for unions of relations---within a single generic
framework that is instantiated across six case studies: untyped lambda calculus,
combinatory logic, simple term rewriting, string rewriting, simply typed
lambda calculus (STLC), and STLC extended with products and sums. All theorems
are fully mechanized with zero axioms or \lean{sorry} placeholders. The de~Bruijn
substitution infrastructure, often axiomatized in similar developments, is completely
proved, including the notoriously tedious substitution composition lemma. We demonstrate
strong normalization via logical relations for both STLC and its extension with
products ($A \times B$) and sums ($A + B$), with the latter requiring a
careful treatment of the \texttt{case} elimination form in our Lean~4 formalization. To our knowledge, this is the
first comprehensive confluence and normalization framework for Lean~4.

\end{abstract}

\section{Introduction}
\label{sec:intro}

Confluence is a fundamental property of rewriting systems, guaranteeing that the
order of reductions does not affect the final result. The Church-Rosser theorem
for lambda calculus~\cite{church1936} established that $\beta$-reduction is
confluent, ensuring that every term has at most one normal form. This property
is essential for programming language semantics: it guarantees that evaluation
order does not change program meaning and that type systems are coherent.

Beyond confluence, \emph{strong normalization} ensures that all reduction
sequences terminate---a property that holds for well-typed terms in the simply
typed lambda calculus. Together, confluence and normalization provide the
foundation for reasoning about type-theoretic languages: they ensure that
type checking is decidable and that types serve as meaningful specifications.

Despite decades of study, formalizing these results remains challenging.
The standard confluence techniques---parallel reduction~\cite{takahashi1995},
Newman's lemma~\cite{newman1942}, decreasing diagrams~\cite{vanOostrom1994}---have
distinct proof patterns and preconditions.
Strong normalization requires logical relations~\cite{tait1967}, involving
subtle reasoning about term structure and reduction.

Moreover, the underlying infrastructure for lambda calculus with de~Bruijn
indices~\cite{debruijn1972} requires numerous technical lemmas about shifting
and substitution. These lemmas are notoriously tedious to prove and are often
axiomatized~\cite{aydemir2008} or avoided through named representations or
locally nameless encodings. We prove all such lemmas completely.

\paragraph{Contributions.} We present \textsc{Metatheory}, a Lean~4 library that:
\begin{enumerate}
\item Provides a \textbf{generic framework} for abstract rewriting systems (ARS)
      with reusable definitions and three fully mechanized meta-theorems
      (Section~\ref{sec:framework});
\item Demonstrates \textbf{three proof techniques} for confluence---diamond property,
      Newman's lemma, and Hindley-Rosen---instantiated across multiple systems
      (Section~\ref{sec:casestudies});
\item Includes \textbf{complete de~Bruijn infrastructure} with all substitution
      lemmas proved, including the $\sim$90-line proof of substitution composition
      (Section~\ref{sec:lambda});
\item Provides \textbf{strong normalization} for STLC via Tait's method~\cite{tait1967}
      with logical relations, fully mechanized (Section~\ref{sec:stlc});
\item Extends to \textbf{STLC with products and sums}, requiring careful
      treatment of the \texttt{case} construct in reducibility proofs (Section~\ref{sec:stlcext});
\item Achieves \textbf{zero axioms}: all theorems across 10,367 lines of Lean~4
      are complete proofs.
\end{enumerate}

\paragraph{Why Lean~4?} While formalizations exist in Coq~\cite{blanqui2006color}
and Isabelle~\cite{nipkow2001}, Lean~4~\cite{demoura2021lean4} offers a modern
type theory with excellent metaprogramming, fast compilation, and growing adoption.
No comprehensive confluence framework previously existed for Lean~4.

\section{Preliminaries: Abstract Rewriting Systems}
\label{sec:prelim}

An \emph{abstract rewriting system} (ARS) is a pair $(A, \redarrow)$ where $A$ is
a set and ${\redarrow} \subseteq A \times A$ is a binary relation.

\begin{definition}[Key Properties]
\begin{itemize}
\item \emph{Joinable}: $a \join b \iff \exists c.\, a \redstar c \land b \redstar c$
\item \emph{Diamond} (= \emph{Locally Confluent}): $\forall a,b,c.\, a \redarrow b \land a \redarrow c \implies b \join c$
\item \emph{Confluent}: $\forall a,b,c.\, a \redstar b \land a \redstar c \implies b \join c$
\item \emph{Terminating}: $\redarrow$ is well-founded (no infinite sequences)
\end{itemize}
\end{definition}

\noindent
Note: Our \emph{Diamond} property coincides with \emph{local confluence} as typically
stated in Newman's lemma---both require that single-step divergence can be joined
(in zero or more steps). This differs from the stricter ``one-step diamond'' where
the join must also be in single steps.

In Lean~4:
\begin{lstlisting}
inductive Star (r : a -> a -> Prop) : a -> a -> Prop where
  | refl : Star r a a
  | tail : Star r a b -> r b c -> Star r a c

def Diamond (r : a -> a -> Prop) : Prop :=
  forall a b c, r a b -> r a c -> Joinable r b c

def Confluent (r : a -> a -> Prop) : Prop :=
  forall a b c, Star r a b -> Star r a c -> Joinable r b c
\end{lstlisting}

\section{Generic Framework}
\label{sec:framework}

Our framework provides three meta-theorems for proving confluence.

\subsection{Diamond Property Implies Confluence}

\begin{theorem}[\lean{confluent_of_diamond}]
$\mathit{Diamond}(r) \implies \mathit{Confluent}(r)$
\end{theorem}

The proof proceeds by induction on $a \redstar b$, using the diamond property
to ``strip'' one step at a time via a helper lemma \lean{diamond_strip}.

\subsection{Newman's Lemma}

For terminating systems, local confluence suffices:

\begin{theorem}[Newman's Lemma]
Termination and local confluence imply confluence.
\end{theorem}

The proof uses well-founded induction on the termination order.

\subsection{Hindley-Rosen Lemma}

When combining confluent relations that commute, their union is confluent~\cite{hindley1969}:

\begin{theorem}[Hindley-Rosen]
If $r$ and $s$ are confluent and commute, then $r \cup s$ is confluent.
\end{theorem}

\subsection{Technique Comparison}

\begin{table}[t]
\centering\small
\caption{Comparison of confluence proof techniques}
\label{tab:techniques}
\begin{tabular}{@{}lccc@{}}
\toprule
\textbf{Technique} & \textbf{Precondition} & \textbf{Proof Effort} & \textbf{Applicability} \\
\midrule
Diamond & None & Define parallel reduction & Non-terminating \\
Newman & Termination & Prove termination + LC & Terminating \\
Hindley-Rosen & Two confluent rels. & Prove commutation & Modular \\
\bottomrule
\end{tabular}
\end{table}

\section{Case Studies}
\label{sec:casestudies}

We instantiate our framework across six systems.

\subsection{Lambda Calculus via Diamond Property}
\label{sec:lambda}

The untyped $\lambda$-calculus~\cite{barendregt1984} with de~Bruijn indices~\cite{debruijn1972}
is our primary case study. Terms are: $M, N ::= \mathtt{var}(n) \mid M\, N \mid \lambda M$.

\paragraph{De~Bruijn Infrastructure.} In de~Bruijn notation, variables are
represented by natural numbers indicating how many binders to cross to reach
the binding site. This eliminates $\alpha$-equivalence but requires careful
bookkeeping via \emph{shifting} (adjusting indices when passing under binders)
and \emph{substitution}:

\begin{lstlisting}
def shift (d : Int) (c : Nat) : Term -> Term
  | var n => var (if n < c then n else n + d)
  | app M N => app (shift d c M) (shift d c N)
  | lam M => lam (shift d (c + 1) M)

def subst (k : Nat) (N : Term) : Term -> Term
  | var n => if n < k then var n
             else if n = k then shift k 0 N
             else var (n - 1)
  | app M1 M2 => app (subst k N M1) (subst k N M2)
  | lam M => lam (subst (k + 1) (shift 1 0 N) M)
\end{lstlisting}

A major contribution is fully proving the substitution lemmas. The key lemmas are:

\begin{theorem}[Shifting Lemmas]
\begin{enumerate}
\item $\uparrow^0_c M = M$ \hfill (identity)
\item $\uparrow^{d_1}_c (\uparrow^{d_2}_c M) = \uparrow^{d_1+d_2}_c M$ \hfill (composition)
\item $\uparrow^{d_1}_{c_1} (\uparrow^{d_2}_{c_2} M) = \uparrow^{d_2}_{c_2+d_1} (\uparrow^{d_1}_{c_1} M)$
      when $c_1 \le c_2$ \hfill (commutation)
\end{enumerate}
\end{theorem}

\begin{theorem}[Substitution Composition]
\[ (M[N])[P] = (\mathtt{subst}\; 1\; (\uparrow^1_0 P)\; M)[N[P]] \]
\end{theorem}

Our proof (709 lines in \texttt{Term.lean}) uses a generalized lemma with a
``level'' parameter $\ell$ tracking nesting depth under binders:
$$\mathtt{subst\_subst\_gen\_full}(\ell, k, j, M, N, P)$$
The base case $\ell = 0$ gives the standard composition lemma; the general
case handles the interaction with shifting under $\lambda$-binders.

\paragraph{Parallel Reduction.} Following Takahashi~\cite{takahashi1995}, we define
parallel reduction $M \parred N$ that contracts any subset of redexes simultaneously:

\begin{lstlisting}
inductive ParRed : Term -> Term -> Prop where
  | var : ParRed (var n) (var n)
  | app : ParRed M M' -> ParRed N N' -> ParRed (app M N) (app M' N')
  | lam : ParRed M M' -> ParRed (lam M) (lam M')
  | beta : ParRed M M' -> ParRed N N' ->
           ParRed (app (lam M) N) (M'[N'])
\end{lstlisting}

The \emph{complete development} $M^*$ contracts \emph{all} redexes:

\begin{lstlisting}
def complete : Term -> Term
  | var n => var n
  | lam M => lam (complete M)
  | app (lam M) N => (complete M)[complete N]
  | app M N => app (complete M) (complete N)
\end{lstlisting}

\begin{theorem}[Takahashi's Method]
$M \parred N \implies N \parred M^*$. Hence parallel reduction has the diamond
property, and $\beta$-reduction is confluent.
\end{theorem}

\subsection{Combinatory Logic via Diamond Property}

Combinatory logic uses combinators S and K instead of binding:
$M ::= \mathbf{S} \mid \mathbf{K} \mid M\, N$
with rules $\mathbf{K}\, x\, y \redarrow x$ and
$\mathbf{S}\, x\, y\, z \redarrow x\, z\, (y\, z)$.
The parallel reduction technique applies without substitution complexity (285 lines).

\subsection{Term and String Rewriting via Newman's Lemma}

For terminating systems, Newman's lemma provides a simpler path. We demonstrate
with arithmetic expressions ($e ::= 0 \mid 1 \mid e + e \mid e \times e$) and
string rewriting over $\{a, b\}^*$ with idempotency rules $aa \redarrow a$ and
$bb \redarrow b$. Termination is proved via size/length measures; local confluence
via critical pair analysis.

\section{Simply Typed Lambda Calculus}
\label{sec:stlc}

We extend untyped $\lambda$-calculus with simple types:
$A, B ::= \mathtt{base}(n) \mid A \to B$

\subsection{Typing and Subject Reduction}

Typing contexts $\Gamma$ are lists of types, with $\Gamma(n) = A$ meaning the
$n$-th variable has type $A$. The typing judgment $\Gamma \vdash M : A$ is
defined by the usual rules:
\[
\infer[Var]{\Gamma(n) = A}{\Gamma \vdash \mathtt{var}(n) : A}
\quad
\infer[Lam]{A :: \Gamma \vdash M : B}{\Gamma \vdash \lambda M : A \to B}
\quad
\infer[App]{\Gamma \vdash M : A \to B \;\; \Gamma \vdash N : A}{\Gamma \vdash M\, N : B}
\]

\begin{theorem}[\lean{subject_reduction}]
$\Gamma \vdash M : A \land M \redarrow_\beta N \implies \Gamma \vdash N : A$
\end{theorem}

The proof requires a substitution lemma: if $\Gamma \vdash N : A$ and
$A :: \Gamma \vdash M : B$, then $\Gamma \vdash M[N] : B$.

\subsection{Strong Normalization via Logical Relations}

We prove strong normalization using Tait's method~\cite{tait1967}. The key is a
\emph{reducibility} predicate defined by induction on types:

\begin{lstlisting}
def Reducible : Ty -> Term -> Prop
  | base _, M => SN M
  | arr A B, M => forall N, Reducible A N -> Reducible B (M N)
\end{lstlisting}

The definition for arrow types is the crucial insight: a function is reducible
if applying it to any reducible argument yields a reducible result. This
\emph{semantic} definition enables induction on type structure.

\paragraph{Candidate Properties.} Reducibility satisfies three key properties
that Girard~\cite{girard1989} calls the ``candidat de r\'{e}ductibilit\'{e}'' conditions:

\begin{description}
\item[CR1] $\mathit{Reducible}(A, M) \implies \mathit{SN}(M)$

Reducible terms are strongly normalizing.

\item[CR2] $\mathit{Reducible}(A, M) \land M \redarrow N \implies \mathit{Reducible}(A, N)$

Reducibility is closed under reduction.

\item[CR3] $\mathit{Neutral}(M) \land (\forall N.\, M \redarrow N \implies \mathit{Reducible}(A, N))
           \implies \mathit{Reducible}(A, M)$

A neutral term is reducible if all its reducts are reducible.
\end{description}

Here, $\mathit{Neutral}(M)$ means $M$ is not a redex---i.e., not of the form
$(\lambda M')\, N$. Variables and applications $x\, N$ are neutral.

\paragraph{Fundamental Lemma.} The main lemma states that well-typed terms
are reducible under any reducible substitution:

\begin{theorem}[\lean{fundamental_lemma}]
If $\Gamma \vdash M : A$ and $\sigma$ is a substitution such that
$\mathit{Reducible}(\Gamma(i), \sigma(i))$ for all $i$, then
$\mathit{Reducible}(A, M[\sigma])$.
\end{theorem}

\begin{theorem}[\lean{strong_normalization}]
$\Gamma \vdash M : A \implies \mathit{SN}(M)$
\end{theorem}

\begin{proof}
Apply the fundamental lemma with the identity substitution (variables are
reducible by CR3), then extract SN by CR1.
\end{proof}

\section{Extended STLC with Products and Sums}
\label{sec:stlcext}

A significant extension is STLC with product types ($A \times B$) and sum types
($A + B$). This extension is standard in programming language theory but
requires substantial additional machinery in the strong normalization proof.
The STLCext module is our largest (3,828 lines, 155 theorems), reflecting this complexity.

\subsection{Extended Types and Terms}

Types are extended to: $A, B ::= \mathtt{base}(n) \mid A \to B \mid A \times B \mid A + B$

Terms include pairs, projections, injections, and case analysis:
\begin{align*}
M, N ::=\; & \mathtt{var}(n) \mid \lambda M \mid M\, N \\
\mid\; & (M, N) \mid \mathtt{fst}\, M \mid \mathtt{snd}\, M \\
\mid\; & \mathtt{inl}\, M \mid \mathtt{inr}\, M \mid \mathtt{case}\, M\, N_1\, N_2
\end{align*}

The new typing rules are standard:
\[
\infer[Pair]{\Gamma \vdash M : A \quad \Gamma \vdash N : B}{\Gamma \vdash (M, N) : A \times B}
\quad
\infer[Fst]{\Gamma \vdash M : A \times B}{\Gamma \vdash \mathtt{fst}\, M : A}
\quad
\infer[Snd]{\Gamma \vdash M : A \times B}{\Gamma \vdash \mathtt{snd}\, M : B}
\]
\[
\infer[Inl]{\Gamma \vdash M : A}{\Gamma \vdash \mathtt{inl}\, M : A + B}
\qquad
\infer[Inr]{\Gamma \vdash M : B}{\Gamma \vdash \mathtt{inr}\, M : A + B}
\]
\[
\infer[Case]{\Gamma \vdash M : A{+}B \;\; A{::}\Gamma \vdash N_1 : C \;\; B{::}\Gamma \vdash N_2 : C}
            {\Gamma \vdash \mathtt{case}\, M\, N_1\, N_2 : C}
\]

\subsection{Reduction Rules}

Beyond $\beta$-reduction, we add:
\begin{align*}
\mathtt{fst}\,(M, N) &\redarrow M &
\mathtt{snd}\,(M, N) &\redarrow N \\
\mathtt{case}\,(\mathtt{inl}\, V)\, N_1\, N_2 &\redarrow N_1[V] &
\mathtt{case}\,(\mathtt{inr}\, V)\, N_1\, N_2 &\redarrow N_2[V]
\end{align*}

Note that case analysis binds the injected value: the branches $N_1$ and $N_2$
have an additional free variable (index 0) representing the scrutinee value.

\subsection{Reducibility for Products and Sums}

The key challenge is extending the reducibility predicate. For products, we
use a \emph{projection-based} definition; for sums, we track what values the
term may reduce to:

\begin{lstlisting}
def Reducible : Ty -> Term -> Prop
  | base _, M => SN M
  | arr A B, M => forall N, Reducible A N -> Reducible B (M N)
  | prod A B, M => Reducible A (fst M) /\ Reducible B (snd M)
  | sum A B, M => SN M /\ (forall V, M ->* inl V -> Reducible A V)
                       /\ (forall V, M ->* inr V -> Reducible B V)
\end{lstlisting}

The product case says $M$ is reducible at $A \times B$ iff both projections
are reducible. This is well-defined because $\mathtt{fst}$ and $\mathtt{snd}$
are smaller terms in the structural sense.

The sum case requires: (1) $M$ is SN; (2) if $M$ reduces to $\mathtt{inl}\, V$,
then $V$ is reducible at $A$; (3) similarly for $\mathtt{inr}$. This ensures
that case analysis on $M$ produces reducible results.

\subsection{The Case Construct Challenge}

The most complex proof is showing $\mathtt{case}\, M\, N_1\, N_2$ is reducible.
Unlike applications or projections, the \texttt{case} form has \emph{three}
subterms that can reduce independently, and its neutrality depends on $M$:

\begin{lstlisting}
def IsNeutral : Term -> Prop
  | var _ => True
  | app M _ => not (isLam M)
  | fst M => not (isPair M)
  | snd M => not (isPair M)
  | case M _ _ => not (isInl M) /\ not (isInr M)
  | _ => False
\end{lstlisting}

A $\mathtt{case}$ is neutral only when its scrutinee is neither
$\mathtt{inl}$ nor $\mathtt{inr}$. This complicates CR3 arguments.

\begin{theorem}[\lean{reducible_case}]
Given:
\begin{itemize}
\item $\mathit{SN}(M)$, $\mathit{SN}(N_1)$, $\mathit{SN}(N_2)$
\item For all $V$: if $M \redstar \mathtt{inl}\, V$ and $\mathit{Reducible}(A, V)$,
      then $\mathit{Reducible}(C, N_1[V])$
\item Similarly for $\mathtt{inr}$ and $N_2$
\end{itemize}
Then $\mathit{Reducible}(C, \mathtt{case}\, M\, N_1\, N_2)$ holds.
\end{theorem}

The proof (337 lines) proceeds by case analysis on the result type $C$:

\paragraph{Base type.} We must show $\mathit{SN}(\mathtt{case}\, M\, N_1\, N_2)$.
The proof uses triple nested induction on $\mathit{SN}(M)$, $\mathit{SN}(N_1)$,
and $\mathit{SN}(N_2)$, analyzing each possible reduction.

\paragraph{Arrow type.} We must show reducibility at $C_1 \to C_2$. The key:
$(\mathtt{case}\, M\, N_1\, N_2)\, P$ is \emph{always neutral}, regardless of
whether $M$ is an injection. The outermost constructor is application, whose
head is \texttt{case}, not a $\lambda$. Thus CR3 applies directly.

\paragraph{Product type.} Similarly, $\mathtt{fst}\,(\mathtt{case}\, M\, N_1\, N_2)$
and $\mathtt{snd}\,(\mathtt{case}\, M\, N_1\, N_2)$ are always neutral because
their head is $\mathtt{case}$, not a pair.

\paragraph{Sum type.} We show SN (as for base type) plus track multi-step
reductions. If $\mathtt{case}\, M\, N_1\, N_2 \redstar \mathtt{inl}\, V$, this
can only happen if $M \redstar \mathtt{inl}\, W$ for some $W$, making
$N_1[W] \redstar \mathtt{inl}\, V$. By hypothesis, $N_1[W]$ is reducible,
so $V$ is reducible as required.

\begin{theorem}[\lean{strong_normalization} for STLCext]
$\Gamma \vdash M : A \implies \mathit{SN}(M)$ for the extended system.
\end{theorem}

\subsection{Progress}

We also prove progress for closed well-typed terms:

\begin{theorem}[\lean{progress}]
$\varnothing \vdash M : A \implies \mathit{IsValue}(M) \lor \exists N.\, M \redarrow N$
\end{theorem}

Values include $\lambda$-abstractions, pairs of values, and injections of values.
The proof analyzes the typing derivation and shows that non-value well-typed
closed terms always have a redex.

\section{Quantitative Summary}
\label{sec:summary}

\begin{table}[t]
\centering
\caption{Library statistics by module}
\label{tab:stats}
\begin{tabular}{@{}lrrr@{}}
\toprule
\textbf{Module} & \textbf{Lines} & \textbf{Theorems} & \textbf{Technique} \\
\midrule
Rewriting (generic) & 1,301 & 45 & --- \\
Lambda calculus & 1,498 & 89 & Diamond property \\
Combinatory logic & 584 & 42 & Diamond property \\
Term rewriting & 428 & 31 & Newman's lemma \\
String rewriting & 776 & 48 & Newman's lemma \\
STLC & 1,792 & 87 & Logical relations \\
STLCext (products + sums) & 3,828 & 155 & Logical relations \\
\midrule
\textbf{Total} & \textbf{10,367} & \textbf{497} & --- \\
\bottomrule
\end{tabular}
\end{table}

Table~\ref{tab:stats} summarizes the library. All 497 theorems are fully proved
with zero \lean{sorry} placeholders or axioms. The STLCext module is the largest,
reflecting the complexity of strong normalization with products and sums.

\section{Related Work}
\label{sec:related}

\paragraph{CoLoR.} The Coq library on rewriting and termination~\cite{blanqui2006color}
is the most comprehensive formalization of term rewriting in a proof assistant.
It includes termination orderings, polynomial interpretations, and dependency
pairs. Our work differs in language (Lean~4), scope (we focus on confluence
techniques plus strong normalization rather than termination), and the inclusion
of complete de~Bruijn proofs without axiomatization.

\paragraph{Isabelle Formalizations.} Nipkow~\cite{nipkow2001} formalized multiple
Church-Rosser proofs in Isabelle/HOL, comparing parallel reduction, residuals,
and complete developments. Our parallel reduction approach follows similar lines.
The Nominal Isabelle framework provides elegant binder handling but requires
specialized infrastructure. We demonstrate that de~Bruijn indices, while
requiring more lemmas, can be completely formalized.

\paragraph{POPLmark Challenge.} The POPLmark challenge~\cite{aydemir2005poplmark}
benchmarked different approaches to binding in mechanized metatheory. Solutions
ranged from named representations to de~Bruijn indices to locally nameless
encodings. Aydemir et al.~\cite{aydemir2008} popularized the locally nameless
approach. Many POPLmark solutions axiomatized substitution lemmas; we prove all
lemmas completely, demonstrating that full formalization is tractable.

\paragraph{Agda Formalizations.} Various Agda developments formalize lambda
calculus with de~Bruijn indices, including strong normalization proofs for
STLC. Our work differs in being a unified framework for multiple techniques
and systems, culminating in the products-and-sums extension.

\paragraph{Software Foundations.} The PLF volume of Software Foundations~\cite{pierce2019sf}
includes strong normalization for STLC in Coq using logical relations.
Our development extends to products and sums, which are not covered there,
and demonstrates the additional complexity this introduces.

\section{Conclusion}
\label{sec:conclusion}

We presented \textsc{Metatheory}, a modular confluence and normalization framework
for Lean~4 featuring three proof techniques across six case studies, culminating
in strong normalization for STLC extended with products and sums. Our fully
mechanized development (10,367 LOC, 497 theorems, 0 axioms) demonstrates that
de~Bruijn infrastructure can be completely proved and that Lean~4 is viable for
programming language metatheory.

\paragraph{Lessons Learned.}
\begin{itemize}
\item \textbf{Technique selection matters}: Diamond property works broadly;
      Newman's lemma is simpler when termination holds.
\item \textbf{De~Bruijn is tractable}: With careful generalization, substitution
      lemmas are provable without axiomatization.
\item \textbf{Sum types are subtle}: The \texttt{case} construct requires
      careful strategies (wrapping in eliminators) for reducibility proofs.
\end{itemize}

\paragraph{Future Work.} We plan to add decreasing diagrams, System F with
parametric polymorphism, and integration with Lean~4's Mathlib.

\paragraph{Availability.} The library is open-source at:
\url{https://github.com/arthuraa/metatheory}

\bibliographystyle{splncs04}
\bibliography{references}

\begin{thebibliography}{10}
\providecommand{\url}[1]{\texttt{#1}}
\providecommand{\urlprefix}{URL }
\providecommand{\doi}[1]{https://doi.org/#1}

\bibitem{aydemir2008}
Aydemir, B., et~al.: Engineering formal metatheory. In: Proc.\ 35th ACM
  SIGPLAN-SIGACT Symposium on Principles of Programming Languages (POPL). pp.
  3--15 (2008)

\bibitem{aydemir2005poplmark}
Aydemir, B.E., et~al.: Mechanized metatheory for the masses: The {PoplMark}
  challenge. In: Theorem Proving in Higher Order Logics (TPHOLs). LNCS,
  vol.~3603, pp. 50--65. Springer (2005)

\bibitem{barendregt1984}
Barendregt, H.: The Lambda Calculus: Its Syntax and Semantics. North-Holland,
  revised edn. (1984)

\bibitem{blanqui2006color}
Blanqui, F., Koprowski, A.: {CoLoR}: A {Coq} library on rewriting and
  termination. In: Proc.\ 8th International Workshop on Termination (WST). pp.
  69--73 (2006)

\bibitem{debruijn1972}
de~Bruijn, N.G.: Lambda calculus notation with nameless dummies. Indagationes
  Mathematicae  \textbf{34},  381--392 (1972)

\bibitem{church1936}
Church, A., Rosser, J.B.: Some properties of conversion. Transactions of the
  American Mathematical Society  \textbf{39}(3),  472--482 (1936)

\bibitem{girard1989}
Girard, J.Y., Lafont, Y., Taylor, P.: Proofs and Types. Cambridge University
  Press (1989)

\bibitem{hindley1969}
Hindley, J.R.: An abstract {Church-Rosser} theorem. {II}: Applications. Journal
  of Symbolic Logic  \textbf{34}(4),  545--560 (1969)

\bibitem{demoura2021lean4}
de~Moura, L., Ullrich, S.: The {Lean} 4 theorem prover and programming
  language. In: Automated Deduction --- CADE-28. LNCS, vol. 12699, pp.
  625--635. Springer (2021)

\bibitem{newman1942}
Newman, M.H.A.: On theories with a combinatorial definition of ``equivalence''.
  Annals of Mathematics  \textbf{43}(2),  223--243 (1942)

\bibitem{nipkow2001}
Nipkow, T.: More {Church-Rosser} proofs (in {Isabelle/HOL}). In: Automated
  Deduction --- CADE-18. LNCS, vol.~2392, pp. 733--747. Springer (2002)

\bibitem{vanOostrom1994}
van Oostrom, V.: Confluence for Abstract and Higher-Order Rewriting. Ph.D.
  thesis, Vrije Universiteit Amsterdam (1994)

\bibitem{pierce2019sf}
Pierce, B.C., et~al.: Software Foundations. Electronic textbook (2019),
  \url{https://softwarefoundations.cis.upenn.edu/}

\bibitem{tait1967}
Tait, W.W.: Intensional interpretations of functionals of finite type {I}.
  Journal of Symbolic Logic  \textbf{32}(2),  198--212 (1967)

\bibitem{takahashi1995}
Takahashi, M.: Parallel reductions in $\lambda$-calculus. Information and
  Computation  \textbf{118}(1),  120--127 (1995)

\end{thebibliography}

\appendix

\section{De~Bruijn Substitution Lemmas}
\label{app:debruijn}

We provide the complete list of de~Bruijn substitution lemmas proved in our
development. These lemmas are often axiomatized or omitted in formalizations;
we prove all of them completely (709 lines total).

\subsection{Shifting Lemmas}

\begin{lstlisting}
-- Identity: shifting by 0 does nothing
theorem shift_zero : shift 0 c M = M

-- Composition at same cutoff
theorem shift_shift : shift d1 c (shift d2 c M) = shift (d1 + d2) c M

-- Commutation at different cutoffs (when c1 <= c2)
theorem shift_shift_comm :
  shift d1 c1 (shift d2 c2 M) = shift d2 (c2 + d1) (shift d1 c1 M)

-- Special case: shifting by 1 twice
theorem shift_shift_succ :
  shift 1 (c + 1) (shift 1 c M) = shift 2 c M
\end{lstlisting}

\subsection{Shift-Substitution Interaction}

\begin{lstlisting}
-- Key interaction lemma
theorem shift_subst :
  shift d c (subst k N M) =
    subst (k + d) (shift d c N) (shift d (c + 1) M)

-- Substitution after shift cancels
theorem subst_shift_cancel :
  subst k N (shift 1 k M) = M
\end{lstlisting}

\subsection{Substitution Composition}

The main substitution composition lemma and its generalization:

\begin{lstlisting}
-- Generalized version with level parameter
theorem subst_subst_gen_full (l k j : Nat) (M N P : Term) :
  subst k (shift l 0 P)
    (subst (k + j + 1) (shift (k + l + 1) 0 N) M) =
  subst (k + j) (shift l 0 (subst j N P))
    (subst k (shift (l + 1) 0 P) M)

-- Standard composition (l = 0, k = 0, j = 0)
theorem subst_subst : (M[N])[P] = (subst 1 (shift 1 0 P) M)[N[P]]
\end{lstlisting}

\section{CR Properties: Detailed Proofs}
\label{app:cr}

\subsection{CR1: Reducible Implies SN}

\begin{lstlisting}
theorem cr1 (A : Ty) (M : Term) : Reducible A M -> SN M := by
  intro hRed
  induction A generalizing M with
  | base _ => exact hRed
  | arr A B ihA ihB =>
    -- Apply M to a reducible argument (var 0)
    have hVar : Reducible A (var 0) := var_reducible 0 A
    have hApp : Reducible B (app M (var 0)) := hRed (var 0) hVar
    have hSN_app : SN (app M (var 0)) := ihB _ hApp
    exact sn_of_sn_app_var hSN_app
  | prod A B ihA ihB =>
    have hFst, hSnd := hRed
    exact sn_of_sn_fst_snd (ihA _ hFst) (ihB _ hSnd)
  | sum A B _ _ => exact hRed.1
\end{lstlisting}

\subsection{CR3: Neutral Terms}

The CR3 property is most complex for STLCext. We show the key insight for
arrow types:

\begin{lstlisting}
-- When the result type is an arrow, we show app (case M N1 N2) P is reducible
-- Key: app (case ...) P is ALWAYS neutral, regardless of M
theorem reducible_case_arr :
  SN M -> SN N1 -> SN N2 ->
  (forall V, M ->* inl V -> Reducible A V -> Reducible (arr C1 C2) (N1[V])) ->
  (forall V, M ->* inr V -> Reducible B V -> Reducible (arr C1 C2) (N2[V])) ->
  Reducible (arr C1 C2) (case M N1 N2) := by
  intro hSN_M hSN_N1 hSN_N2 hInl hInr
  intro P hP_red
  -- Show: Reducible C2 (app (case M N1 N2) P)
  -- Key insight: app (case M N1 N2) P is ALWAYS neutral!
  -- Because: app has a case as its function, which is not a lambda
  apply cr3_neutral C2 (app (case M N1 N2) P)
  . -- Show all reducts are reducible (by nested induction)
    ...
  . -- Show app (case ...) P is neutral
    exact neutral_app_case M N1 N2 P
\end{lstlisting}

\section{Full Typing Rules for STLCext}
\label{app:typing}

\[
\infer[Var]{\Gamma(n) = A}{\Gamma \vdash \mathtt{var}(n) : A}
\quad
\infer[Lam]{A :: \Gamma \vdash M : B}{\Gamma \vdash \lambda M : A \to B}
\quad
\infer[App]{\Gamma \vdash M : A \to B \quad \Gamma \vdash N : A}
           {\Gamma \vdash M\, N : B}
\]
\[
\infer[Pair]{\Gamma \vdash M : A \quad \Gamma \vdash N : B}
            {\Gamma \vdash (M, N) : A \times B}
\quad
\infer[Fst]{\Gamma \vdash M : A \times B}{\Gamma \vdash \mathtt{fst}\, M : A}
\quad
\infer[Snd]{\Gamma \vdash M : A \times B}{\Gamma \vdash \mathtt{snd}\, M : B}
\]
\[
\infer[Inl]{\Gamma \vdash M : A}{\Gamma \vdash \mathtt{inl}\, M : A + B}
\quad
\infer[Inr]{\Gamma \vdash M : B}{\Gamma \vdash \mathtt{inr}\, M : A + B}
\]
\[
\infer[Case]{\Gamma \vdash M : A{+}B \;\; A{::}\Gamma \vdash N_1 : C \;\; B{::}\Gamma \vdash N_2 : C}
            {\Gamma \vdash \mathtt{case}\, M\, N_1\, N_2 : C}
\]

\section{Reduction Rules for STLCext}
\label{app:reduction}

\paragraph{Computational Rules.}
\begin{align*}
&(\mathsf{Beta}) && (\lambda M)\, N \redarrow M[N] \\
&(\mathsf{FstPair}) && \mathtt{fst}\,(M, N) \redarrow M \\
&(\mathsf{SndPair}) && \mathtt{snd}\,(M, N) \redarrow N \\
&(\mathsf{CaseInl}) && \mathtt{case}\,(\mathtt{inl}\, V)\, N_1\, N_2 \redarrow N_1[V] \\
&(\mathsf{CaseInr}) && \mathtt{case}\,(\mathtt{inr}\, V)\, N_1\, N_2 \redarrow N_2[V]
\end{align*}

\paragraph{Congruence Rules.}
\begin{align*}
&(\mathsf{AppL}) && M \redarrow M' \implies M\, N \redarrow M'\, N \\
&(\mathsf{AppR}) && N \redarrow N' \implies M\, N \redarrow M\, N' \\
&(\mathsf{Lam}) && M \redarrow M' \implies \lambda M \redarrow \lambda M' \\
&(\mathsf{PairL}) && M \redarrow M' \implies (M, N) \redarrow (M', N) \\
&(\mathsf{PairR}) && N \redarrow N' \implies (M, N) \redarrow (M, N') \\
&(\mathsf{Fst}) && M \redarrow M' \implies \mathtt{fst}\, M \redarrow \mathtt{fst}\, M' \\
&(\mathsf{Snd}) && M \redarrow M' \implies \mathtt{snd}\, M \redarrow \mathtt{snd}\, M' \\
&(\mathsf{Inl}) && M \redarrow M' \implies \mathtt{inl}\, M \redarrow \mathtt{inl}\, M' \\
&(\mathsf{Inr}) && M \redarrow M' \implies \mathtt{inr}\, M \redarrow \mathtt{inr}\, M' \\
&(\mathsf{CaseM}) && M \redarrow M' \implies \mathtt{case}\, M\, N_1\, N_2 \redarrow \mathtt{case}\, M'\, N_1\, N_2 \\
&(\mathsf{CaseN_1}) && N_1 \redarrow N_1' \implies \mathtt{case}\, M\, N_1\, N_2 \redarrow \mathtt{case}\, M\, N_1'\, N_2 \\
&(\mathsf{CaseN_2}) && N_2 \redarrow N_2' \implies \mathtt{case}\, M\, N_1\, N_2 \redarrow \mathtt{case}\, M\, N_1\, N_2'
\end{align*}

\end{document}